\documentclass[submission,copyright,creativecommons]{eptcs}
\providecommand{\event}{TARK 2015} 
\usepackage{breakurl}             

\title{Parameterized Complexity Results for a Model of Theory of Mind Based on Dynamic Epistemic Logic\thanks{This research has been carried out in the context of the first author's master's thesis \cite{vandepol2015}.}}
\author{Iris van de Pol\thanks{Supported by Gravitation Grant 024.001.006 of the Language in Interaction Consortium from the Netherlands Organization for Scientific Research.}
\institute{Institute for Logic, Language\\and Computation}
\institute{University of Amsterdam}
\email{i.p.a.vandepol@uva.nl}
\and
Iris van Rooij
\institute{Donders Institute for Brain,\\Cognition, and Behaviour}
\institute{Radboud University}
\email{\quad i.vanrooij@donders.ru.nl}
\and
Jakub Szymanik\thanks{Supported by the Netherlands Organisation for Scientific Research Veni Grant NWO-639-021-232.}
\institute{Institute for Logic, Language\\and Computation}
\institute{University of Amsterdam}
\email{j.k.szymanik@uva.nl}
}
\def\titlerunning{Parameterized Complexity Results for a Model of Theory of Mind Based on DEL}
\def\authorrunning{Iris van de Pol, Iris van Rooij \& Jakub Szymanik}

\usepackage[utf8]{inputenc}
\usepackage{mathtools}
\usepackage{hyperref}
\usepackage{mathtools}
\usepackage{tikz}
\usetikzlibrary{decorations.pathreplacing}
\usetikzlibrary{arrows}
\usepackage{graphicx}
\usepackage{enumitem}
\usepackage{amsmath}
\usepackage{amssymb}
\usepackage{multirow}
 \usepackage{framed}
 \usepackage{enumitem}
 \usepackage{ marvosym }
 \usepackage[boxed,noend,figure]{algorithm2e}
 \usepackage{booktabs}
\usepackage{centernot}

\usepackage{amsthm}

\newcommand{\set}[1]{\{#1\}}

\newcommand{\w}{\wedge}

\newcommand{\se}{\ensuremath{\subseteq}}

\newcommand{\ra}{\rightarrow}

\newcommand{\Sm}{\Sigma}

\renewcommand{\phi}{\varphi}

\newtheorem{definition}{\sc Definition}
\numberwithin{definition}{section}
\newtheorem{theorem}{\sc Theorem}

\newtheorem{proposition}[theorem]{\sc Proposition}

\newcommand{\mtext}[1]{\text{\normalfont #1}}
\newcommand{\pre}{\text{\textit{pre}}}
\newcommand{\post}{\text{\textit{post}}}
\newcommand{\varp}[0]{\textit{var}($\phi$)}

 \definecolor{darkgreen}{rgb}{0.0, 0.5, 0.5}

\newenvironment{myquote}{\begin{center}
    \begin{minipage}{.95\linewidth}}{\end{minipage}\end{center}}

\newcommand{\probdefd}[3]{
  \begin{myquote}
  \begin{framed}
    #1
    
    \textit{Instance:} #2
    
    \textit{Question:} #3
    \end{framed}
  \end{myquote}
}

\newcommand{\probdefp}[4]{
    \begin{myquote}
    \begin{framed}
    #1
    
    \textit{Instance:} #2

	\textit{Parameter:} #3
    
    \textit{Question:} #4
  
    \end{framed}
  \end{myquote} 
}

\makeatletter
\DeclareRobustCommand{\rvdots}{%
  \vbox{
    \baselineskip4\p@\lineskiplimit\z@
    \kern-\p@
    \hbox{.}\hbox{.}\hbox{.}
  }}
\makeatother

\newcommand{\dbu}[0]{{\sc DBU}}
\newcommand{\gdbu}[0]{{\sc DBU}}

\newcommand{\ogdbu}[0]{$\set{o}$-{\sc DBU}}
\newcommand{\agdbu}[0]{$\set{a}$-{\sc DBU}}

\newcommand{\pgdbu}[0]{$\set{p}$-{\sc DBU}}

\newcommand{\eugdbu}[0]{$\set{e,u}$-{\sc DBU}}

\newcommand{\acpgdbu}{\ensuremath{\set{a,c,p}\text{-{\sc DBU}}}}
\newcommand{\cfpugdbu}{\ensuremath{\set{c,f,p,u}\text{-{\sc DBU}}}}

\newcommand{\aefopgdbu}{\ensuremath{\set{a,e,f,o,p}\text{-{\sc DBU}}}}
\newcommand{\cefopgdbu}{\ensuremath{\set{c,e,f,o,p}\text{-{\sc DBU}}}}

\newcommand{\copugdbu}{\ensuremath{\set{c,o,p,u}\text{-{\sc DBU}}}}

\newcommand{\afopugdbu}{\ensuremath{\set{a,f,o,p,u}\text{-{\sc DBU}}}}
\newcommand{\acefogdbu}{$\set{a,c,e,f,o}$-{\sc DBU}}
\newcommand{\acfougdbu}{$\set{a,c,f,o,u}$-{\sc DBU}}

\newcommand{\sat}[0]{{\sc Sat}}

\newcommand{\pwsat}{\ensuremath{\set{k}\text{-{\sc WSat\mtext{[2CNF]}}}}}
\newcommand{\mclique}{\ensuremath{\set{k}\text{-{\sc Multicolored Clique}}}}
\newcommand{\paranp}{\mtext{para-NP}}

\newcommand{\SBs}{\{}%
\newcommand{\SEs}{\}}%

\renewcommand{\P}{\text{\normalfont P}}

\newcommand{\np}{\text{\normalfont NP}}

\newcommand{\FPT}{\text{\normalfont FPT}}

\newcommand{\W}[1]{\text{\normalfont W[#1]}}

\newcommand{\pspace}{\ensuremath{\mtext{PSPACE}}}

\newcommand{\tqbf}{\ensuremath{\mtext{\sc TQBF}}}

\newcommand{\LB}{\mathcal{L}_B}
\newcommand{\A}{\mathcal{A}}

\newcommand{\AAA}{\mathsf{A}}

\newcommand{\EEE}{\mathcal{E}}
 
\newcommand{\LLL}{\mathcal{L}}

\newcommand{\nn}{\mathbb{N}}

\newcommand{\CNF}{\mtext{\sc CNF}}

\DeclareRobustCommand{\VAN}[3]{#3}

\begin{document}
\maketitle

\begin{abstract}
In this paper we introduce a 
computational-level model of theory of mind (ToM) based on dynamic epistemic logic (DEL), and we analyze its computational complexity. The model is a special case of DEL model checking. We provide a parameterized complexity analysis, considering several aspects of DEL
 (e.g., number of agents, size of preconditions, etc.) as parameters.
We show that model checking for DEL is PSPACE-hard, also when restricted to single-pointed models and S5 relations, thereby solving an open problem in the literature.
Our approach is aimed at formalizing current intractability 
claims in the cognitive science literature regarding computational models of ToM. 
\end{abstract}

\section{Introduction}

Imagine that you are in love. You find yourself at your desk, but you cannot stop your mind from wandering off. What is she thinking about right now? And more importantly, is she thinking about you 
and  does she know that you 
are thinking about her?
Reasoning about other people's knowledge, belief and desires, we do it all the time. 
For instance, 
in trying to conquer the love of one's life, to 
stay one step ahead of one's enemies,  
or   when we lose our friend in a crowded place and we
 find them by imagining where they would look for us.
This capacity is known as theory of mind (ToM) and it is widely 
studied in  various fields (see, e.g., 
\cite{bolander2014,brauner2013,frith2001,nichols2003,perea2012,premack1978,verbrugge2009,wellman2001}).

We seem to use ToM on a daily {basis and} 
many cognitive scientists consider
it to be ubiquitous in social interaction \cite{apperly2010}. 
At the same time, however, it is also widely believed that computational cognitive models 
of ToM  are intractable, i.e., that ToM involves
solving problems that humans are not capable of solving
(cf.~\cite{apperly2010,haselager1997,levinson2006,zawidzki2013}).
This seems to imply a contradiction between theory 
and practice: 
on the one hand we seem to be capable of ToM, while
on the other hand, our theories tell us that this is impossible.
Dissolving this paradox is a {critical} step in enhancing theoretical understanding of ToM.

The question arises what it means for a computational-level model\footnote{In cognitive science, often Marr's \cite{marr1982} tri-level distinction between computational-level (``what is the nature of the problem being solved?"), algorithmic-level (``what is the algorithm used for solving the problem?"), and implementational-level (``how is the algorithm physically realized?") is used to distinguish different levels of computational cognitive explanations. In this paper, we will focus on computational-level models of ToM and their computational complexity.}
of cognition to be intractable.
 When looking more closely at these intractability claims
 regarding ToM, it is not clear
what these researchers mean exactly,
nor whether they mean the same thing.
In theoretical computer science and logic there are a variety of tools to make precise 
claims about the level of complexity of a certain 
problem.
 In cognitive science, however, this is a different story.
 With the exception of a few researchers, cognitive scientists 
 do not  tend to  specify formally what it means for a theory to be intractable. 
 This makes it often very difficult to assess the validity 
 of the various claims in the literature about 
 which theories are {tractable} and which are not.

\sloppypar
In this paper we 
adopt the \textit{Tractable-cognition thesis} (see~\cite{vanrooij2008}) that states that people have 
limited resources for cognitive processing and  
human cognitive capacities are confined to those 
that can be realized using a realistic amount of time.\footnote{There is general consensus in the cognitive science community that computational intractability is a undesirable feature of cognitive computational models, putting the cognitive plausibility of such models into question \cite{cherniak1986,frixione2001,gigerenzer2008,vanrooij2008,tsotsos1990}. There are diverging opinions about how cognitive science should deal with this issue (see, e.g., \cite{Chater2006,gigerenzer2008,vanrooij2014,vanrooij2012}). It is beyond the scope of this paper to discuss this in detail. In this paper we adopt the parameterized complexity approach as described in \cite{vanrooij2008}.}
More specifically we  adopt the \textit{FPT-cognition thesis} \cite{vanrooij2008} that states that computationally plausible computational-level cognitive theories are limited to the class of input-output mappings that are fixed-parameter tractable for one or more input-parameters that can  be assumed to be small in practice. 
To be able to make more precise claims about the
(in)tractability of ToM  we introduce a computational-level model of ToM based on dynamic epistemic logic (DEL), and we analyze its computational complexity.
The model we present is a special case of DEL model checking.
Here we include an informal description of the model.\footnote{
We pose the model in the form of a decision problem, as this is convenient for purposes of our complexity analysis. Even though ToM may be more intuitively modeled by a search problem, the complexity of the decision problem gives us lower bounds on the complexity of such a search problem, and therefore suffices for the purposes of our paper. 
} 
The kind of situation that we want to be able to model, is that of
an observer that observes one or more agents in an initial situation. The observer then witnesses actions that change the situation and the observer updates their knowledge about the mental states of the agents in the new situation. Such a set up is often found in experimental tasks, where subjects are asked to reason about the mental states of agents in a situation that they are presented.

\probdefd{ \dbu{} (informal) -- {\sc Dynamic Belief Update}
}{A representation of an initial situation, a sequence of actions -- observed by an observer -- and a (belief) statement~$\phi$ of interest. 
}{Is the (belief) statement~$\phi$ true  in the 
situation resulting from the initial situation and the observed
actions?
}

We prove that \dbu{} is \pspace{}-complete. \pspace{}-completeness was already shown by Aucher and Schwarzentruber \cite{aucher2013} for 
DEL model checking 
in general. They considered unrestricted relations and multi-pointed event models. Since their proof does not 
hold for the special case of DEL model checking that we consider, we propose  an alternative proof. 
 Our proof solves positively the open question in \cite{aucher2013} whether model checking for DEL restricted to S5 relations and single-pointed models is \pspace{}-complete. 
Bolander, Jensen and Schwarzentruber \cite{bolander2015} independently considered an almost identical special case of DEL model checking (there called the plan verification problem). They also prove \pspace{}-completeness for the case restricted to 
single-pointed models, but their proof does not settle whether hardness holds 
even when the problem is restricted to S5 models.

Furthermore, we investigate how the different aspects (or parameters, see Table~\ref{table:parameters}) of our model influence its complexity. We prove that for most combinations of parameters \dbu{} is fp-intractable and for one case we prove
fp-tractability. See Figure~\ref{fig:overview} for an overview of the results. 

Besides the parameterized complexity results for DEL model checking that we present, the main conceptual contribution of this paper is that it bridges cognitive science and logic, by using DEL to model ToM (cf.~\cite{isaac2014,verbrugge2009}). By doing so, the paper provides the means to make more precise statements about the (in)tractability of ToM.

The paper is structured as follows.  In Section~\ref{section:preliminaries} we introduce 
basic definitions from dynamic epistemic logic and parameterized complexity theory. 
Then, in Section~\ref{section:complevelmodel} we introduce a formal 
description of our computational-level model and we discuss the particular choices that we make.
Next, in Section~\ref{complexresults} we present our (parameterized) 
complexity results. Finally, in Section~\ref{discussion} we discuss the implications of our results for the understanding of ToM.

\section{Preliminaries}
\label{section:preliminaries}
\subsection{Dynamic Epistemic Logic}
\label{DELformal}

Dynamic epistemic logic is a particular kind of modal logic 
(see
 \cite{ditmarsch2008,Ben11}),
where the modal operators are interpreted in terms of belief 
or knowledge. First, we define epistemic models, which are  Kripke models
 with an accessibility relation for every agent~$a \in \A$,
instead of just one accessibility relation.

\begin{definition}[Epistemic model] 
Given a finite set~$\A$ of agents and a finite set~$P$ of propositions, an epistemic model is a tuple~$(W,R,V)$ where
\begin{itemize}
\item $ $ $W$ is a non-empty set of worlds;
\item $ $ $R$
is a function that assigns to every agent $a \in \A$ a binary relation $R_a$ on $W$; and
\item $ $ $V$
is a valuation function from  $W \times P$ into $\set{0,1}$.
\end{itemize}

\end{definition}

The accessibility relations~$R_a$ can be read as follows: for worlds $w,v \in W$, $w R_a v$ 
means ``in world $w$, agent~$a$ considers world~$v$ possible.”

\begin{definition}[(Multi and single-)pointed
  epistemic model]
A pair~$(M, W_d)$ consisting of an epistemic model~$M=(W,R,V)$  and a non-empty set of designated worlds~$W_d \se W$ is called a pointed epistemic model. A pair $(M, W_d)$ is called a single-pointed model when~$W_d$ is a singleton, and a multi-pointed epistemic model when~$|W_d| > 1$. By a slight abuse of notation, for~$(M, \set{w})$, we also write~($M,w$). 
\end{definition}

We consider the usual restrictions on relations in epistemic models and event models, such as KD45 and S5 (see \cite{ditmarsch2008}).
In KD45 models, all relations are transitive, Euclidean and serial, and in S5 models all relations are transitive, reflexive and symmetric.

We define the following language for epistemic models. We use the modal 
belief operator~$B$, where for each agent~$a \in \A$,~$B_a \phi$ is 
interpreted as ``agent~$a$ believes (that)~$\phi$".
 
\begin{definition}[Epistemic language]
The  
language~$\LB$ over~$\A$ and~$P$ 
 is given by the following definition, where~$a$ ranges over~$\A$ and~$p$ over~$P$:
\[ \phi :: = p \ | \ \neg \phi \ | \ (\phi \wedge \phi) \ | \ B_a \phi. \]
We will use the following standard abbreviations,~$\top := p \vee \neg p, \bot := \neg \top, \phi \vee \psi := \neg(\neg \phi \wedge \neg \psi)$, $\phi \ra \psi := \neg \phi \vee \psi$, $\hat{B_a}:= \neg B_a \neg \phi$.
\end{definition} 

The semantics 
 for this language is defined as follows.

\begin{definition}[Truth in a (single-pointed) epistemic model]
Let~$M=(W,R,V)$ be an epistemic model,~$w \in W$,~$a \in \A$, and~$\phi, \psi \in \LLL_B$. We define~$M, w \models  \phi$ inductively as follows:
\begin{center}
\begin{tabular}{l l l}
$M,w \models p$ & iff & $V(w,p)=1$
\\
$M,w \models \neg \phi$ & iff & not $M, w \models \phi$
\\
$M,w \models  (\phi \w \psi)$ & iff &  $M, w \models \phi$ and $M, w \models \psi$
\\
$M,w \models B_a \phi$ & iff & for all $v$ with $wR_av$: $M, v \models \phi$
\\
\end{tabular}
\end{center}
When~$M,w \models \phi$, we say that~$\phi$ is true in~$w$ or~$\phi$ is \textit{satisfied} in~$w$. 
\end{definition}

\begin{definition}[Truth in a multi-pointed epistemic model]
Let $(M,W_d)$ be a multi-pointed epistemic model, $a \in \A$, and $\phi\in \LB$. $M, W_d \models  \phi$ is defined as follows:
\begin{center}
 \begin{tabular}{l l l}
 $M,W_d \models \phi$ & iff & $M,w \models \phi$ for all $w \in W_d$
 \\
 \end{tabular}
\end{center}
\end{definition}

Next we define event models.


\begin{definition}[Event model]
An event model is a tuple $\mathcal{E} = (E, Q, \textit{\text{pre}}, \textit{\text{post}})$, where 
$E$ is a non-empty finite set of events;
$Q$ is a function that assigns to every agent $a \in \A$ a binary relation $R_a$ on $W$;
\textit{pre} is a function from $E$ into $\LLL_B$ that assigns to each event a precondition, which can be any formula in~$\LLL_B$; and
\textit{post} is a function from $E$ into $\LLL_B$ that assigns to each event a postcondition. Postconditions 
are conjunctions of propositions and their negations (including $\top$ and $\bot$).
\end{definition}

\begin{definition}[(Multi and single-)pointed  event model / action]
A pair~$(\EEE, E_d)$ consisting of an event model~$\EEE=(E,Q,\pre,\post)$  and a non-empty set of designated events $E_d \se E$ is called a pointed event model. A pair $(\EEE, E_d)$ is called a single-pointed event model when~$E_d$ is a singleton, and a multi-pointed event model when~$|E_d| > 1$. 
We will also refer to~$(\EEE, E_d)$ as an action.
\end{definition}

We define the notion of a product update, that is used to update
epistemic models with actions 
\cite{baltag1998}.

\begin{definition}[Product update]
The  product update of the state~$(M, W_d)$ with the action~$(\EEE, E_d)$ is defined as the state~$(M, W_d) \otimes (\EEE, E_d) = ((W', R', V'), W'_d)$ where
\begin{itemize}
\item $W' = \set{(w,e) \in W \times E\ ;\ M,w \models \text{pre}(e)}$;
\item $R'_a = \set{((w,e), (v,f)) \in W' \times W' \ ; \ wR_av \text{ and } eQ_af}$;
\item $V'(p) = 1$ iff either  $(M,w \models p$ and $ post(e) \not\models \neg p)$ or  $post(e) \models p$; and
\item $W'_d = \set{(w,e) \in W' \ ;\  w \in W_d \text{ and } e \in E_d}$.
\end{itemize}
\end{definition}

Finally, we define when actions are applicable in a state.

\begin{definition}[Applicability]
An action $(\EEE, E_d)$ is applicable in state $(M,W_d)$ if there is some $e \in E_d$ and some $ w \in W_d$ such that $M,w \models pre(e)$. We define applicability for a sequence of
actions inductively. The empty sequence, consisting of no actions, is always applicable. A sequence $a_1, \dots, a_k$ of actions is applicable in a state $(M,W_d)$ if (1) the sequence $a_1, \dots, a_{k-1}$ is applicable in $(M,W_d)$ and (2) the action $a_k$ is applicable in the state $(M,W_d) \otimes a_1 \otimes \dotsm \otimes a_{k-1}$.
\end{definition}

\subsection{Parameterized Complexity Theory}
\label{sectionparameter}

We introduce some basic concepts of {parameterized} complexity theory. For a more detailed introduction we refer to textbooks on the topic
\cite{downey1999,downey2013,flum2006,niedermeier2006}.

\begin{definition}[Parameterized problem]
Let~$\Sm$ be a finite alphabet.
A \emph{parameterized problem}~$L$ (over~$\Sm$) is a subset of~${\Sm^*} \times \nn$. For an \emph{instance}~$(x,k)$, we call~$x$ the \emph{main part} and~$k$ the \emph{parameter}. 

\end{definition}

The complexity class \FPT{}, which stands for fixed-parameter tractable,
is the direct analogue of the class \P{} in classical complexity. 
Problems in this class are considered efficiently solvable, because
the non-polynomial-time complexity inherent in the problem is confined to the parameter and in effect the problem is efficiently solvable even for large input sizes, provided 
that the value of the parameter is relatively small.


\begin{definition}[Fixed-parameter tractable / {t}he class FPT]
Let~$\Sm$ be a finite alphabet. 
\begin{enumerate}
\item
An algorithm~$\AAA$ with input~$(x,k) \in \Sm \times \nn$ runs in  \emph{fpt-time} if there exists a computable function $f$ and a polynomial $p$ such that for all~$(x, k) \in \Sm \times \nn$, the running time of~$\AAA$ on~$(x,k)$ is at most
\[f(k) \cdot p(|x|).\]

\sloppypar
 Algorithms that run in fpt-time are called \emph{fpt-algorithms}.
\item
\sloppypar
A parameterized problem~$L$ is \emph{fixed-parameter tractable} if there is an fpt-algorithm that decides $L$.
 \emph{FPT} denotes the class of all fixed-parameter tractable problems.
\end{enumerate}
\end{definition}

Similarly to classical complexity, parameterized complexity also offers
a hardness framework to give evidence that (parameterized) problems are not
fixed-parameter tractable. The following notion of reductions plays an important
role in this framework.

\begin{definition}[Fpt-reduction] Let~$L \se \Sm \times \nn$ and~$L' \se \Sm' \times \nn$ be two parameterized problems.
An \emph{fpt-reduction} from ~$L$ to~$L'$ is a mapping~$R: \Sm \times \nn \ra \Sm' \times \nn $ from  instances of~$L$ to instances of~$L'$ such that there is a computable function~$g: \nn \ra \nn$ such that for all~$(x,k) \in \Sm \times \nn$:
\begin{enumerate}
\item $ (x',k')=R(x,k)$ is a yes-instance of $L'$ if and only if~$(x,k)$~is a yes-instance of $L$; 
\item $R$ is computable in fpt-time; and
\item  $k' \le g(k)$.
\end{enumerate}
\end{definition}

{Another important part of the hardness framework is the parameterized
intractability class W[1].
To characterize this class}, we consider the following parameterized problem.

\probdefp{\pwsat}{A 2CNF propositional formula $\phi$ and an integer~$k$.}{$k$.}{Is there an assignment~$\alpha:$ \varp $\ \ra \set{0,1}$,  that sets~$k$ variables in \textit{var}($\phi$) to true, that satisfies~$\phi$?}

The class \W{1} consists of all parameterized problems that can be fpt-reduced
to \pwsat{}.
A parameterized problem is hard for W[1] if all problems in W[1] can be
fpt-reduced to it. It is widely believed that W[1]-hard problems are not
fixed-parameter tractable \cite{downey2013}.
Another parameterized intractability class, that can be used in a similar way,
is the class para-NP.
The class \paranp{} consists of all parameterized problems that can be solved by a nondeterministic fpt-algorithm.
To show para-NP-hardness, it suffices to show that \gdbu{} is NP-hard for a constant value of the parameters \cite{flum2003}.
Problems that are para-NP-hard are not fixed-parameter tractable,
unless~$\mtext{P} = \mtext{NP}$ \cite[Theorem~2.14]{flum2006}.


\section{Computational-level Model of Theory of Mind}
\label{section:complevelmodel}

Next we present a formal description of  our computational-level model.
Our aim is to capture, in a qualitative way, the kind
of reasoning that is necessary to be able to engage in ToM. 
Arguably, the essence of ToM is
the attribution of mental states to another person, based
on observed behavior, and to predict and explain this behavior
in terms of those mental states. 
The aspect of ToM that we aim to formalize with our model is the attribution of mental states. 
There is a wide range of different kinds of mental states such as
epistemic, emotional and motivational states.
In our model we 
focus  on epistemic states, in particular on belief.

To be cognitively plausible, our model needs to be able 
to capture a wide range of (dynamic) situations,
where all kinds of actions can occur, not just actions that change 
beliefs (epistemic actions), but also actions that change 
the  state of the world (ontic actions). This is why, following 
Bolander and Andersen \cite{bolander2011}, we use postconditions in the product update of DEL 
(in addition to preconditions). 

Furthermore, we want to model the (internal) perspective of the
observer (on the situation). Therefore, the god 
perspective, also called the perfect external approach 
by Aucher \cite{aucher2010} -- that is inherent to 
 single-pointed epistemic models -- 
 will not suffice for all cases 
that we want to be able to model.
This perfect external approach supposes that the modeler 
is an omniscient observer that is perfectly aware of the 
actual state of the world
 and the epistemic situation 
(what is going on in the minds of the agents).
The cognitively plausible observers that we are
interested in here will not have infallible knowledge in
 many situations.
They are often not  able to distinguish the actual world from
other possible worlds, because they are uncertain about the
 facts in the world and the mental states
of the agent(s) that they observe.
That is why, again following Bolander and Andersen \cite{bolander2011},
we allow for multi-pointed epistemic models
  (in addition to
single-pointed models), which can model the uncertainty of an
observer, by representing their perspective
as a set of worlds.
How to represent the internal or fallible perspective
of an agent in epistemic models is a conceptual problem
that has not been settled yet in the DEL-literature. 
There have been several proposals to deal with this
 (see, e.g., \cite{aucher2010, degremont2014, szymanik2011}).

Also, since we do not assume that agents are perfectly knowledgeable, we allow the possibility of modeling false beliefs of the observers and agents, by using KD45 models (rather than S5 models).
%
Even though KD45 models present an idealized form
of belief (with perfect introspection and logical omniscience), we argue that at least to some extent 
they are cognitively plausible, and that 
 therefore, for the purpose of this paper,
it suffices to focus on KD45 models.
Our complexity results (which we present in the next section)
 do not depend on this choice; they hold for DBU restricted to KD45 models and restricted to S5 models, and also for the unrestricted case.

We define our computational-level model of ToM as follows. 

\probdefd{ \dbu{} (formal) -- {\sc Dynamic Belief Update}
}
{A set of propositions P, and set of Agents~$\A$. An initial state~$s_o$, where~$s_o =  ((W,V,R), W_d)$ is a pointed epistemic model. An applicable sequence of actions~$a_1, ...,a_k$, where~$a_j=( (E, Q, pre, post), E_d)$ is a pointed event model. A formula~$\phi \in \LLL_B$.
}{Does~$s_o \otimes a_1 \otimes ... \otimes a_k \models \phi$?
}

The model can  be naturally used to formalize ToM tasks that are employed in psychological experiments. The classical ToM task
 that is used by (developmental) psychologists is the false belief 
 task \cite{baron1985,perner1983}. The DEL-based
 formalization of the false belief task by  Bolander \cite{bolander2014} 
 can be seen as an instance of \gdbu{}.
 For more details on how \gdbu{} can be used to model ToM tasks, we refer to \cite{vandepol2015}.

\section{Complexity Results}
\label{complexresults}

\subsection{PSPACE-completeness}
\label{sectiongeneralcomplex}

We show that \gdbu{} is PSPACE-complete.
For this, we consider the decision problem \tqbf{}. 
This problem is \pspace{}-complete \cite{stockmeyer1973}.
\probdefd{
  \tqbf{}}{A quantified Boolean formula
    $\varphi = Q_1 x_1 Q_2 x_2 \dotsc Q_m x_m. \psi$.}{Is~$\varphi$ true?
}

\begin{theorem}
\label{thm4}
\gdbu{} is \pspace{}-hard.
\end{theorem}

\begin{proof}
To show \pspace{}-hardness we specify a polynomial-time reduction $R$ from \tqbf{} to \gdbu.
Let $\psi$ be a Boolean formula.
First, we sketch the general idea behind the reduction. 
We use the reduction to list 
all possible assignments to \textit{var}$(\psi)$.
To do this we use groups of worlds (which are~$R_a$-equivalence classes) to represent particular 
truth assignments. 
Each group consists of a string of worlds 
that are fully connected by equivalence relation~$R_a$.
Except for the first world in the string, 
all worlds represent a true variable~$x_i$ (under a particular assignment). 

We give an example of such a group of worlds  that represents assignment $\alpha = \set{x_1 \mapsto \text{T}, x_2 \mapsto \text{F}, x_3 \mapsto \text{T}, x_4 \mapsto \text{T}, x_5 \mapsto \text{F}, x_6 \mapsto \text{T}}$. 
Each world has a reflexive loop for every agent, which we 
leave out for the sake of presentation.  
More generally, in all our drawings we replace each relation~$R_a$ with a minimal~$R_a'$ whose transitive reflexive closure is equal to~$R_a$. \hspace{-8pt}
\scalebox{0.8}{
\begin{tikzpicture}
  \tikzstyle{dnode}=[inner sep=1pt,outer sep=1pt,draw,circle,minimum width=9pt]
  \node[dnode] (w0) at (0,0) {};
  \fill (w0) circle [radius=2pt];
\end{tikzpicture}} 
marks the designated world. Since all relations are reflexive, we draw relations as lines (leaving out arrows at the end). 

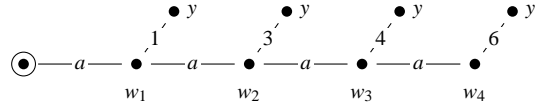
\begin{center}
\begin{tikzpicture}[node distance=0.3cm]
  \tikzstyle{dnode}=[inner sep=1pt,outer sep=1pt,draw,circle,minimum width=9pt]
  \tikzstyle{nnode}=[inner sep=1pt,outer sep=1pt,circle,minimum width=9pt]
  \tikzstyle{label-edge}=[midway,fill=white, inner sep=1pt]
  \node[dnode, label=below:{\scriptsize }] (w0) at (0,0) {};
  \fill (w0) circle [radius=2pt];
  \node[nnode, label=below:{\scriptsize $w_1$}] (w1) at (1.5,0) {};
  \fill (w1) circle [radius=2pt];
   \node[nnode, label={[label distance=-0.15cm]right:{\scriptsize $y$}}] (w11) at (2,0.7) {};
  \fill (w11) circle [radius=2pt];
  \node[nnode, label=below:{\scriptsize $w_2$}] (w2) at (3,0) {};
  \fill (w2) circle [radius=2pt];
   \node[nnode, label={[label distance=-0.15cm]right:{\scriptsize $y$}}] (w22) at (3.5,0.7) {};
  \fill (w22) circle [radius=2pt];
  \node[nnode, label=below:{\scriptsize $w_3$}] (w3) at (4.5,0) {};
  \fill (w3) circle [radius=2pt];
      \node[nnode, label={[label distance=-0.15cm]right:{\scriptsize $y$}}] (w33) at (5.0,0.7) {};
  \fill (w33) circle [radius=2pt];
    \node[nnode, label=below:{\scriptsize $w_4$}] (w4) at (6,0) {};
  \fill (w4) circle [radius=2pt];
    \node[nnode, label={[label distance=-0.15cm]right:{\scriptsize $y$}}] (w44) at (6.5,0.7) {};
  \fill (w44) circle [radius=2pt];
  \draw[-] (w0) -- (w1) node [label-edge] {\scriptsize $a$};
  \draw[-] (w1) -- (w2) node [label-edge] {\scriptsize $a$};
  \draw[-] (w1) -- (w11) node [label-edge] {\scriptsize $1$};
  \draw[-] (w2) -- (w3) node [label-edge] {\scriptsize $a$};
  \draw[-] (w2) -- (w22) node [label-edge] {\scriptsize $3$};
  \draw[-] (w3) -- (w33) node [label-edge] {\scriptsize $4$}; 
  \draw[-] (w3) -- (w4) node [label-edge] {\scriptsize $a$};
  \draw[-] (w4) -- (w44) node [label-edge] {\scriptsize $6$};
\end{tikzpicture}
\end{center}

We refer to worlds~$w_1, \dots, w_4$ as the \textit{bottom worlds} of this group. 
If a bottom world 
has an~$R_i$ relation to a world that makes proposition~$y$ true, we say that it represents variable~$x_i$.

The reduction makes sure 
that in the final updated model (the model that results 
from updating the initial state with the actions 
-- which are specified by the reduction)
each possible truth assignment to the variables in $\psi$ 
will be represented by  a group of worlds. 
Between the different groups, there are no $R_a$-relations (only $R_i$-relations for ~$1 \leq i \leq m$). 
By `jumping' from one group (representing a particular truth assignment)
to another group with relation~$R_i$, the truth value
 of variable~$x_i$ 
 can be set to true or false. 
 We can now 
translate a quantified Boolean formula into a corresponding formula of~$\LLL_B$ by mapping every universal quantifier~$Q_i$ to~$B_i$ and every existential quantifier~$Q_j$ to~$\hat{B}_j$.

To illustrate how this reduction works, we give an example. 
Figure~\ref{fig:thm1} shows the final updated model for a quantified Boolean formula with variables~$x_1$ and~$x_2$. In this model there are four groups of worlds: $\set{w_1,w_2,w_3}$, $\set{w_4, w_5}$, $\set{w_6, w_7}$ and $\set{w_8}$. Worlds~$w_1, \dots, w_8$  {are what} we refer to as the bottom worlds. {The gray worlds and edges can be considered a byproduct of the reduction; 
they have no particular function.}

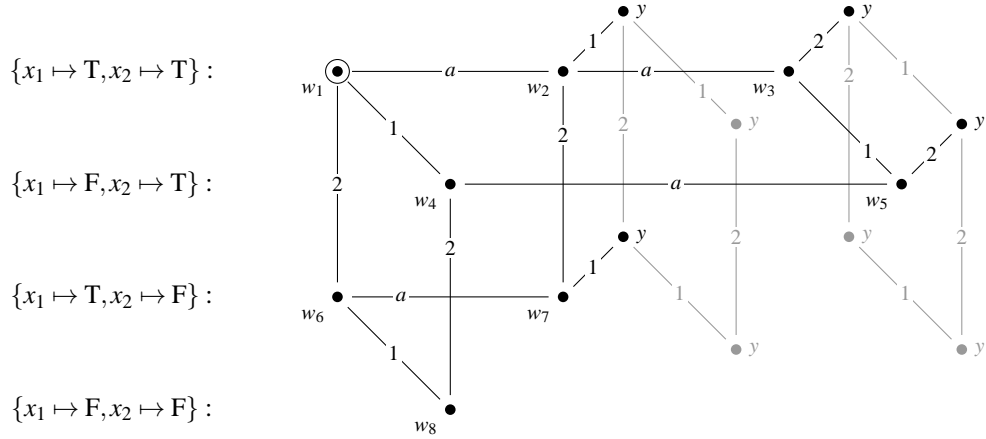
\begin{figure*}
\begin{center}
\begin{tikzpicture}[node distance=0.3cm]
  \tikzstyle{dnode}=[inner sep=1pt,outer sep=1pt,draw,circle,minimum width=9pt]
  \tikzstyle{nnode}=[inner sep=1pt,outer sep=1pt,circle,minimum width=9pt]
  \tikzstyle{label-edge}=[midway,fill=white, inner sep=1pt]
  \tikzstyle{grijs}=[black!40]
  \tikzstyle{donkergrijs}=[black!40]
  \node[dnode, label={[label distance=-0.15cm]below left:{\scriptsize $w_1$}}] (w1) at (0,0) {};
  \fill (w1) circle [radius=2pt];
  \node[nnode, label={[label distance=-0.15cm]below left:{\scriptsize $w_2$}}] (w2) at (3,0) {};
  \fill (w2) circle [radius=2pt];
  \node[nnode, label={[label distance=-0.15cm]right:{\scriptsize $y$}}] (w2b) at (3.8,0.8) {};
  \fill (w2b) circle [radius=2pt];
  \node[nnode, label={[label distance=-0.15cm]below left:{\scriptsize $w_3$}}] (w3) at (6,0) {};
  \fill (w3) circle [radius=2pt];
  \node[nnode, label={[label distance=-0.15cm]right:{\scriptsize $y$}}] (w3b) at (6.8,0.8) {};
  \fill (w3b) circle [radius=2pt];
  \node[nnode, label={[label distance=-0.15cm]below left:{\scriptsize $w_4$}}] (w1') at (1.5,-1.5) {};
  \fill (w1') circle [radius=2pt];
  \node[nnode, label={[label distance=-0.15cm, grijs]right:{\scriptsize $y$}}] (w2b') at (5.3,-0.7) {};
  \fill[grijs] (w2b') circle [radius=2pt];
  \node[nnode, label={[label distance=-0.15cm]below left:{\scriptsize $w_5$}}] (w3') at (7.5,-1.5) {};
  \fill (w3') circle [radius=2pt];
  \node[nnode, label={[label distance=-0.15cm]right:{\scriptsize $y$}}] (w3b') at (8.3,-0.7) {};
  \fill (w3b') circle [radius=2pt];
  \node[nnode, label={[label distance=-0.15cm]below left:{\scriptsize $w_6$}}] (w1'') at (0,-3) {};
  \fill (w1'') circle [radius=2pt];
  \node[nnode, label={[label distance=-0.15cm]below left:{\scriptsize $w_7$}}] (w2'') at (3,-3) {};
  \fill (w2'') circle [radius=2pt];
  \node[nnode, label={[label distance=-0.15cm]right:{\scriptsize $y$}}] (w2b'') at (3.8,-2.2) {};
  \fill (w2b'') circle [radius=2pt];
  \node[nnode, label={[label distance=-0.15cm, grijs]right:{\scriptsize $y$}}] (w3b'') at (6.8,-2.2) {};
  \fill[grijs] (w3b'') circle [radius=2pt];
  \node[nnode, label={[label distance=-0.15cm]below left:{\scriptsize $w_8$}}] (w1''') at (1.5,-4.5) {};
  \fill (w1''') circle [radius=2pt];
  \node[nnode, label={[label distance=-0.15cm, grijs]right:{\scriptsize $y$}}] (w2b''') at (5.3,-3.7) {};
  \fill[grijs] (w2b''') circle [radius=2pt];
  \node[nnode, label={[label distance=-0.15cm, grijs]right:{\scriptsize $y$}}] (w3b''') at (8.3,-3.7) {};
  \fill[grijs] (w3b''') circle [radius=2pt];
  \draw[-,grijs] (w2b'') -- (w2b''') node [label-edge] {\scriptsize $1$};
  \draw[-,grijs] (w3b'') -- (w3b''') node [label-edge] {\scriptsize $1$};
  \draw[-,grijs] (w2b) -- (w2b') node [label-edge, near end] {\scriptsize $1$};
  \draw[-,grijs] (w3b) -- (w3b') node [label-edge] {\scriptsize $1$};
  \draw[-,grijs] (w2b) -- (w2b'') node [label-edge] {\scriptsize $2$};
  \draw[-,grijs] (w2b') -- (w2b''') node [label-edge] {\scriptsize $2$};
  \draw[-,grijs] (w3b) -- (w3b'') node [label-edge, near start] {\scriptsize $2$};
  \draw[-,grijs] (w3b') -- (w3b''') node [label-edge] {\scriptsize $2$};
  \draw[-] (w1) -- (w2) node [label-edge] {\scriptsize $a$};
  \draw[-] (w2) -- (w3) node [label-edge, pos=0.35] {\scriptsize $a$};
  \draw[-] (w1') -- (w3') node [label-edge] {\scriptsize $a$};
  \draw[-] (w1'') -- (w2'') node [label-edge, near start] {\scriptsize $a$};
  \draw[-] (w1) -- (w1') node [label-edge] {\scriptsize $1$};
  \draw[-] (w1'') -- (w1''') node [label-edge] {\scriptsize $1$};
  \draw[-] (w3) -- (w3') node [label-edge, near end] {\scriptsize $1$};
  \draw[-] (w1) -- (w1'') node [label-edge] {\scriptsize $2$};
  \draw[-] (w1') -- (w1''') node [label-edge, near start] {\scriptsize $2$};
  \draw[-] (w2) -- (w2'') node [label-edge, near start] {\scriptsize $2$};
  \draw[-] (w2) -- (w2b) node [label-edge] {\scriptsize $1$};
  \draw[-] (w2'') -- (w2b'') node [label-edge] {\scriptsize $1$};
  \draw[-] (w3) -- (w3b) node [label-edge] {\scriptsize $2$};
  \draw[-] (w3') -- (w3b') node [label-edge] {\scriptsize $2$};
  \node[] at (-3,0) {\small $\SBs x_1 \mapsto \text{T}, x_2 \mapsto \text{T} \SEs:$};
  \node[] at (-3,-1.5) {\small $\SBs x_1 \mapsto \text{F}, x_2 \mapsto \text{T} \SEs:$};
  \node[] at (-3,-3) {\small $\SBs x_1 \mapsto \text{T}, x_2 \mapsto \text{F} \SEs:$};
  \node[] at (-3,-4.5) {\small $\SBs x_1 \mapsto \text{F}, x_2 \mapsto \text{F} \SEs:$};
\end{tikzpicture}
\end{center}
\caption{Example for the reduction in the proof of Theorem 1; a final updated model for a quantified Boolean formula with variables~$x_1$ and~$x_2$.}
\label{fig:thm1}
\end{figure*}

We represent variable~$x_1$ by~$\hat{B}_1 y$ and variable~$x_2$ by~$\hat{B}_2 y$. Then, in the model above, checking whether $\exists x_1 \forall x_2. x_1 \vee x_2$ is true can be done by checking whether \allowbreak formula~$\hat{B}_1 B_2 ( \hat{B}_a \hat{B}_1 y \vee \hat{B}a \hat{B}_2 y)$ is true, which is indeed the case.
Also, checking whether $\forall x_1 \forall x_2. x_1 \vee x_2$ is true can be done by checking whether~$B_1 B_2 (\hat{B}_a \hat{B}_1 y \vee \hat{B}a \hat{B}_2 y)$ is true, which is not the case.

 \sloppypar 
Now, we continue with the formal details.
Let $\phi = Q_1 x_1 \dots Q_m x_m. \psi$ 
be a quantified Boolean formula with quantifiers~$Q_1,\dots,Q_m$ and \textit{var}$(\psi) = \set{x_1, \dots , x_m}$.
 We define the following polynomial-time computable mappings. For~$1 \le i \le m$, let $[x_i] =  \hat{B}_i y$, and  
\[ [Q_i] = 
\begin{dcases*}
B_i & if $Q_i = \forall$ \\
\hat{B}_i & if $Q_i = \exists $. \\
\end{dcases*}
\] 

  Formula~$[\psi]$ is the adaptation of formula $\psi$  
 where 
 every occurrence of~$x_i$ in~$\psi$ is replaced by~$\hat{B}_a[x_i]$. Then~$[\phi] = [Q_1] \dots [Q_m] [\psi]$.
We formally specify the reduction~$R$.
 We let $R(\phi) = (P, \A, s_0, a_1, \dots, a_{m},  [\phi])$, where:

\begin{itemize}
\item $P = \set{y}$,
$\A = \set{a, 1, \dots, m}$

\item 
\begin{tabular}{r l}
$s_0 = $
&
\hspace{-6pt}
\parbox[c]{\hsize}{
\begin{tikzpicture}[node distance=0.3cm]
  \tikzstyle{dnode}=[inner sep=1pt,outer sep=1pt,draw,circle,minimum width=9pt]
  \tikzstyle{nnode}=[inner sep=1pt,outer sep=1pt,circle,minimum width=9pt]
  \tikzstyle{label-edge}=[midway,fill=white, inner sep=1pt]
  \node[dnode, label=below:{\scriptsize }] (w0) at (0,0) {};
  \fill (w0) circle [radius=2pt];
  \node[nnode, label=below:{\scriptsize }] (w1) at (1.5,0) {};
  \fill (w1) circle [radius=2pt];
   \node[nnode, label={[label distance=-0.15cm]right:{\scriptsize $y$}}] (w11) at (2.1,0.7) {};
  \fill (w11) circle [radius=2pt];
  \node[nnode, label=below:{\scriptsize }] (w2) at (3,0) {};
  \fill (w2) circle [radius=2pt];
   \node[nnode, label={[label distance=-0.15cm]right:{\scriptsize $y$}}] (w22) at (3.6,0.7) {};
  \fill (w22) circle [radius=2pt];
  \node[] (w3) at (4.2,0) {$\cdots$};
    \node[nnode, label=below:{\scriptsize }] (w4) at (5.4,0) {};
  \fill (w4) circle [radius=2pt];
    \node[nnode, label={[label distance=-0.15cm]right:{\scriptsize $y$}}] (w44) at (6,0.7) {};
  \fill (w44) circle [radius=2pt];
  \draw[-] (w0) -- (w1) node [label-edge] {\scriptsize $a$};
  \draw[-] (w1) -- (w2) node [label-edge] {\scriptsize $a$};
  \draw[-] (w1) -- (w11) node [label-edge] {\scriptsize $1$};
  \draw[-] (w2) -- (w3) node [label-edge] {\scriptsize $a$};
  \draw[-] (w2) -- (w22) node [label-edge] {\scriptsize $2$};
  \draw[-] (w3) -- (w4) node [label-edge] {\scriptsize $a$};
  \draw[-] (w4) -- (w44) node [label-edge] {\scriptsize $m$};
\end{tikzpicture}
} \\
\end{tabular}
\end{itemize}

All relations in~$s_0, a_1,\dots,a_m$ are equivalence relations. Note that all worlds in~$s_0, a_1,\dots,a_m$ have reflexive loops for 
all agents. We omit all reflexive loops for the sake of readability.

\begin{itemize}
\item

\begin{tabular}{r l}
$a_1 = $
&
\parbox[c]{\hsize}{
\begin{tikzpicture}[node distance=1cm]
  \tikzstyle{dnode}=[inner sep=1pt,outer sep=1pt,draw,circle,minimum width=9pt]
  \tikzstyle{nnode}=[inner sep=1pt,outer sep=1pt,circle,minimum width=9pt]
  \tikzstyle{label-edge}=[midway,fill=white, inner sep=1pt]
  \node[dnode,label=south:{\scriptsize $e_1 : \langle \top, \top  \rangle$}] (w1) at (0,0) {};
  \fill (w1) circle [radius=2pt];
  \node[nnode,label=south:{\scriptsize $e_2 : \langle \neg \hat{B}_1 y \vee y, \top \rangle$}] (w2) at (2,0) {};
  \fill (w2) circle [radius=2pt];
  \draw[-] (w1) -- (w2) node [label-edge] {\scriptsize $1$};
\end{tikzpicture}
} \\
\end{tabular}

\indent \vdots

\item

\begin{tabular}{r l}
$a_m = $
\parbox[c]{\hsize}{
\begin{tikzpicture}[node distance=1cm]
  \tikzstyle{dnode}=[inner sep=1pt,outer sep=1pt,draw,circle,minimum width=9pt]
  \tikzstyle{nnode}=[inner sep=1pt,outer sep=1pt,circle,minimum width=9pt]
  \tikzstyle{label-edge}=[midway,fill=white, inner sep=1pt]
  \node[dnode,label=south:{\scriptsize $e_1 : \langle \top, \top \rangle$}] (w1) at (0,0) {};
  \fill (w1) circle [radius=2pt];
  \node[nnode,label=south:{\scriptsize $e_2 : \langle \neg \hat{B}_m y \vee y , \top \rangle$}] (w2) at (2,0) {};
  \fill (w2) circle [radius=2pt];
  \draw[-] (w1) -- (w2) node [label-edge] {\scriptsize $m$};
\end{tikzpicture}
} \\
\end{tabular}
\end{itemize}

We show that~$ \phi \in $ \tqbf{} if and only if~$R(\phi) \in$ \gdbu{}.
We prove that for all~$1 \le i \le m+1$ the following claim holds. For any assignment~$\alpha$ to the variables~$x_1,\dots,x_{i-1}$ and any bottom world~$w$ of a group that agrees with~$\alpha$, 
the formula $Q_i x_i \dots Q_m x_m. \psi$ is true under~$\alpha$ if and only if $[Q_i]\dots [Q_m] [\psi]$ is true in world~$w$. In the case for $i=m+1$, this refers to the formula~$[\psi]$.

We start with the case for~$i = m+1$. We show that the claim holds. Let~$\alpha$ be  any assignment to the variables~$x_1, \dots,x_m$, and let~$w$ be any bottom world of a group~$\gamma$ that represents~$\alpha$. Then, by construction of~$[\psi]$, we know that~$\psi$ is true under~$\alpha$ if and only if~$[\psi]$ is true in~$w$. 

Assume that the claim holds for~$i = j + 1$. We show that then the claim also holds for~$i = j$. Let~$\alpha$ be  any assignment to the variables~$x_1, \dots, x_{j-1}$ and let~$w$ be a bottom world of a group that agrees with~$\alpha$. We show that the formula~$Q_j \dots Q_m.\psi$ is true under~$\alpha$ if and only if~$[Q_j] \dots [Q_m][\psi]$ is true in~$w$. 

First, assume that~$Q_j \dots Q_m.\psi$ is true under~$\alpha$.
Consider the case where~$Q_j = \forall$. Then for both assignments~$\alpha ' \supseteq \alpha$ to the variables~$x_1, \dots, x_j$, formula~$Q_{j+1} \dots Q_m.\psi$ is true under~$\alpha '$. Now, by assumption, we know that for any bottom world~$w'$ of a group that agrees with~$\alpha$ --
 so in particular for all bottom worlds~$w'$ that are $R_j$-reachable from~$w$ -- formula $[Q_{j+1}] \dots [Q_m][\psi]$ is true in~$w'$. Since~$[Q_{j}] = B_j$, this means that~$[Q_{j}] \dots [Q_m][\psi]$ is true in~$w$.
The case where~$Q_j = \exists$ is analogous.

Next, assume that~$Q_j \dots Q_m.\psi$ is not true under~$\alpha$.
Consider the case where~$Q_j = \forall$. Then there is some 
assignment~$\alpha ' \supseteq \alpha$ to the variables~$x_1, \dots, x_j$, such 
that~$Q_{j+1} \dots Q_m.\psi$ is not true under~$\alpha '$.
Now, by assumption, we know that for any bottom world~$w'$ 
of a group that agrees with~$\alpha$ -- so in particular 
for some bottom world~$w'$ that is  $R_j$-reachable 
from~$w$ -- formula~$[Q_{j+1}] \dots [Q_m][\psi]$ is not true in~$w'$. 
Since~$[Q_{j}] = B_j$, this means that~$[Q_{j}] \dots [Q_m][\psi]$ 
is not true in~$w$.
The case where~$Q_j = \exists$ is analogous.

Hence, the claim holds for the case that~$i = j$. Now, by induction, the claim holds for the case that~$i = 1$, and hence it follows that~$ \phi \in $ \tqbf{} if and only if~$R(\phi) \in$ \gdbu{}.
Since this reduction runs in polynomial time, we can conclude that \gdbu{} is \pspace{}-hard.
\end{proof}

\begin{theorem}
\label{thm5}
\gdbu{} is \pspace{}-complete.
\end{theorem}

\begin{proof}
In order to show PSPACE-membership for the problem DBU, we can modify the polynomial-space algorithm given by Aucher and Schwarzentruber \cite{aucher2013}. Their algorithm works for the problem of checking whether a given (single-pointed) epistemic model makes a given DEL-formula true, where the formula contains event models that can be multi-pointed, but that have no postconditions. In order to make the algorithm work for multi-pointed epistemic models, we can simply call the algorithm several times, once for each of the designated worlds. Also, a modification is needed to deal with postconditions. The algorithm checks the truth of a formula by inductively calling itself for subformulas. In order to deal with postconditions, only the case where the formula is a propositional variable needs to be modified. This modification is rather straightforward. For more details, we refer to~\cite{vandepol2015}.
\end{proof}

\subsection{Parameterized Complexity Results}
\label{sectionparametercomplex}

Next, we provide a parameterized complexity analysis of \gdbu{}.

\subsubsection{Parameters for \gdbu{}}

 We consider the following parameters for \gdbu{}.
For each subset~$\kappa \subseteq \{ a,c,e,f,o,p,u \}$ we consider
the parameterized variant~$\kappa \mtext{-\gdbu}$ of \gdbu{},
where the parameter is the sum of the values for the elements
of~$\kappa$ as specified in Table~\ref{table:parameters}.
For instance, the problem~\agdbu{} is
parameterized by the number of agents.
Even though technically speaking there is only one parameter,
we will refer to each of the elements of~$\kappa$ as parameters.

For the modal depth of a formula we count the maximum number of
nested occurrences of operators~$B_a$.
Formally, we define the modal depth~$d(\varphi)$ of a formula~$\varphi$
(in $\LLL_B$) recursively as follows.
\[ d(\varphi) = \begin{dcases*}
0& if $\varphi = p\in P$ is a proposition;\\
\max \SBs d(\varphi_1), d(\varphi_2) \SEs &if $\varphi = \varphi_1 \wedge \varphi_2$;\\
d(\varphi_1) &if $\varphi = \neg \varphi_1$;\\
1 + d(\varphi_1) &if $\varphi = B_a \varphi_1$.\\
\end{dcases*} \]

For the size of a formula we count the number of occurrences of
propositions and logical connectives.
Formally, we define the size~$s(\varphi)$ of a formula~$\varphi$
(in $\LLL_B$) recursively as follows.
\[ s(\varphi) = \begin{dcases*}
1&if $\varphi = p \in P$ is a proposition;\\
1 + s(\varphi_1) + s(\varphi_2) &if $\varphi = \varphi_1 \wedge \varphi_2$;\\
1 + s(\varphi_1) &if $\varphi = \neg \varphi_1$;\\
1 + s(\varphi_1) &if $\varphi = B_a \varphi_1$.\\
\end{dcases*} \]

\begin{table}[tb]
  \centering
  \begin{tabular}{@{}c@{\quad}l@{}}
    \textbf{Param.} & \textbf{Description} \\
    \toprule
    $a$ & number of agents \\
    \midrule
    $c$ & maximum size of the preconditions \\
    \midrule
    $e$ & maximum number of events in the event models \\
    \midrule
    $f$ & size of the formula \\
    \midrule
    $o$ & modal depth of the formula, \\
    & i.e., the order parameter \\
    \midrule
    $p$ & number of propositions in $P$ \\
    \midrule
    $u$ & number of actions, i.e., the number of updates \\
    \bottomrule
  \end{tabular}
  \caption{Overview of the different parameters for \gdbu{}.}
  \label{table:parameters}
\end{table}

\subsubsection{Intractability Results}

In the following, we 
show fixed-parameter intractability for
 several parameterized versions of \gdbu{}.
We will mainly use the parameterized complexity classes W[1]
and para-NP to show intractability, i.e., we will show hardness
for these classes.
Note that we could additionally use the class para-PSPACE \cite{flum2003} to give
stronger intractability results.
For instance, the proof of Theorem~\ref{thm4} already shows
that \pgdbu{} is para-PSPACE hard, since the reduction in this proof
uses  a constant number of propositions.
However, since in this paper we are mainly interested in the
border between fixed-parameter tractability and intractability,
we will not focus on the subtle differences in the degree of
intractability, and restrict ourselves to showing W[1]-hardness
and para-NP-hardness.
This is also the reason why we will not show membership for
any of the (parameterized) intractability classes; showing hardness  suffices to
indicate intractability.
For the following proofs we use the well-known satisfiability problem \sat{}
for propositional formulas. The problem \sat{} is \np{}-complete \cite{cook1971,levin1973}. Moreover, hardness for \sat{} holds even when restricted to propositional formulas that are in 3\CNF.

\begin{proposition}
\label{proposition6}
\acefogdbu{} is \paranp{}-hard.
\end{proposition}

\begin{proof}
To show \paranp{}-hardness, we specify a polynomial-time 
reduction $R$ from {\sat{}} 
 to \gdbu{}, where parameters~$a$, $c$, $e$, $f$, and~$o$ 
have constant values. 
Let $\phi$ be a propositional formula with 
\textit{var}$(\phi) = \set{x_1, \dots , x_m}$. 
Without loss of generality we assume that~$\phi$ is a 
3\CNF{}  formula with clauses $c_1$ to $c_l$. 
   
The general idea behind this reduction is that we use the worlds
in the final updated model (that results from updating the initial state with the 
actions -- which are specified by the reduction) to {list} 
all possible assignments to  \textit{var}$(\phi)$, by setting 
the propositions (corresponding to the variables in 
\textit{var}$(\phi)$) to true and false accordingly. Then checking 
whether formula~$\phi$ is satisfiable can be done by checking 
whether~$\phi$ is true in any of the worlds. 
Actions~$a_1$ to~$a_m$ are used to 
create a corresponding world for each possible assignment 
to  \textit{var}$(\phi)$. Furthermore, 
to keep the formula that we check in the 
final updated model of constant size, 
we sequentially check the truth of each clause~$c_i$ and 
encode whether the clauses are true with an
additional  variable~$x_{m+1}$. 
This is done by actions~$a_{m+1}$ to~$a_{m+l}$. In the final updated 
model, variable~$x_{m+1}$ will only be true in a world, 
if it makes clauses~$c_1$ to~$c_l$ true, i.e., 
if it makes formula~$\phi$ true.

For more details, we refer to~\cite{vandepol2015}.
\end{proof}

\begin{proposition}
\label{prop11}
\cefopgdbu{} is \paranp{}-hard.
\end{proposition}

\begin{proof}
To show para-NP-hardness, we specify a polynomial-time 
reduction $R$ from {\sat{}} 
 to \gdbu{}, where parameters~$c$, $e$, $f$, $o$, and~$p$ 
have constant values.
Let $\phi$ be a propositional formula with 
\textit{var}$(\phi) = \set{x_1, \dots , x_m}$.
The general  idea behind this reduction 
is similar to the reduction in the proof of Theorem~\ref{thm4}.
Again we  use groups of worlds to represent particular assignments
to the variables in ~$\phi$. Here, there is only relation~$R_b$
between the different groups. 
Furthermore, 
to keep the formula that we check in the 
final updated model of constant size, 
we sequentially check the truth of each clause~$c_i$ and 
encode whether the clauses are true with an
additional variable~$z$. 
This is done by actions~$a_{m+1}$ to~$a_{m+l}$. 
Action~$a_{m+j}$ (corresponding to clause~$j$) marks each group of worlds 
(which represents a particular assignment to the variables in~$\phi$)  that `satisfies' clauses~$1$ to~$j$. (This marking happens by means of an~$R_c$-accessible world where~$z$ is true.) Then, in the final updated model, there will only be such a marked group if all clauses, and hence the whole formula, is satisfiable.
 
For more details, we refer to~\cite{vandepol2015}.
\end{proof}

\begin{proposition}
\label{prop12}
\aefopgdbu{} is \paranp{}-hard.
\end{proposition}

\begin{proof}
To show para-NP-hardness, we specify a polynomial-time reduction $R$ from
\sat{} to \gdbu, where 
parameters~$a$, $e$, $f$, $o$ and~$p$  
have constant values. 
Let $\phi$ be a propositional formula with~\textit{var}$(\phi) = \set{x_1, \dots , x_m}$. 
The reduction is based on the same principle as the one used in the  
proof of Proposition~\ref{prop11}. To keep the number of agents constant, we 
use a different construction to represent 
the variables in~\textit{var}$(\phi)$.
We encode the variables by a string of worlds that are 
connected by alternating relations~$R_a$ and~$R_b$.

 Furthermore, we keep the size of the formula 
 (and consequently the modal depth of the formula) constant by encoding
 the satisfiability  of the formula with a single proposition.
 We do this   
 by adding an extra action~$a_{m+1}$. 
 Action~$a_{m+1}$ makes sure that each group of worlds that represents a satisfying assignment for
 the given formula, will have an~$R_c$ relation from a world that 
 is~$R_b$-reachable from the designated world to a world where
 proposition~$z^*$ is true.  

For more details, we refer to~\cite{vandepol2015}.
\end{proof}

We consider the following parameterized problem, that we will use in our proof of Proposition~\ref{prop9}.
{This problem is \W{1}-complete \cite{fellows2009}.}

\probdefp{\mclique{}}{A graph $G$, and a vertex-coloring $c: V(G) \ra \set{1,2,\dots,k}$ for $G$.}{$k$.}{Does $G$ have a clique of size $k$ including vertices of all $k$ colors? That is, are there $v_1, \dots, v_k \in V(G)$ such that for all $1 \le i < j \le k: \set{ v_i, v_j} \in E(G)$ and $c(v_i) \neq c(v_j)$?}

\begin{proposition}
\label{prop9}
\acfougdbu{} is \W{1}-hard.
\end{proposition}

\begin{proof}
\sloppypar
We specify an fpt-reduction $R$ from \mclique{} to \acfougdbu{}. 
Let $(G,c)$ be an instance of \mclique{}, where~$G=(N,E)$. {The general idea  behind 
this reduction is that
we use the worlds in the model to list all $k$-sized 
subsets of the vertices in the graph with~$k$ different colors, where each individual world
represents a particular $k$-subset of vertices in the graph (with $k$ different colors). 
Then we encode (in the model) the existing edges 
between these nodes (with particular color endings), and in the 
final updated model we check whether there is a world
corresponding to a $k$-subset of vertices that is pairwise
fully connected with edges. 
This is only the case when~$G$
has a~$k$-clique with~$k$ different colors. }

For more details, we refer to~\cite{vandepol2015}.
\end{proof}

\begin{proposition}
\label{thm13}
\copugdbu{} is \W{1}-hard.
\end{proposition}

\begin{proof}
\sloppypar
We specify the following fpt-reduction~$R$ from
\pwsat{} 
to \copugdbu{}.  
We sketch the general idea behind the reduction.
Let $\phi$ 
be a propositional formula with \textit{var}$(\phi) = \set{x_1, \dots , x_m}$. 
Then let~$\phi'$ be {the formula obtained from~$\phi$, by replacing each
occurrence of~$x_i$ with~$\neg x_i$}.
We note that~$\phi$ is satisfiable by some assignment~$\alpha$ that 
sets~$k$ variables to true if and only if~$ \phi'$ is satisfiable by 
some assignment~$\alpha'$ that sets~$m-k$ variables to true, i.e.,
that sets~$k$ variables to false.
We use the reduction to list all possible assignments 
to \textit{var}$(\phi')$ = \textit{var}$(\phi)$ that set~$m - k$ variables to true.
We represent each possible assignment 
to \textit{var}$(\phi)$ that sets~$m - k$ variables to true 
as a group of worlds, like in the proof of Theorem~\ref{thm4}. (In fact, due to the details 
of the reduction,
in the final updated model, there
will be several identical groups of worlds for 
each of these assignments).

For more details, we refer to~\cite{vandepol2015}.
\end{proof}

\begin{proposition}
\label{prop15}
\afopugdbu{} is \W{1}-hard.
\end{proposition}
\begin{proof}
\sloppypar
We specify the following fpt-reduction~$R$ from \pwsat{}
to \afopugdbu{}. 
We modify the reduction in the proof of Proposition~\ref{thm13} to keep the values of parameters~$a$ and $f$ 
constant. After  these modifications, the value of parameter~$c$ 
will no longer be constant.
To keep the number of agents constant, 
we use the same {strategy} as in the reduction in the proof of 
Proposition~\ref{prop12}, {where variables~$x_i,\dots,x_m$ 
are represented by strings of worlds with alternating relations~$R_b$ and~$R_a$.}
Just like in the proof of Proposition~\ref{prop12}, the size of the formula 
 (and consequently the modal depth of the formula) 
 is kept constant by encoding
 the satisfiability  of the formula with a single proposition.
 Then each group of worlds that represents a satisfying assignment for
 the given formula, will have an~$R_c$ relation from a world that 
 is~$R_b$-reachable from the designated world to a world where
 proposition~$z^*$ is true.

For more details, we refer to~\cite{vandepol2015}.
\end{proof}

\subsubsection{Tractability Results}
\label{section:tractabilityresults}

 Next, we turn to a case that is fixed-parameter tractable.

\begin{theorem}
\eugdbu{} is fixed-parameter tractable.
\end{theorem}

\begin{proof}
We present the following fpt-algorithm that runs in time~$e^u \cdot p(| 
x|)$, for some polynomial $p$, where~$e$ is the maximum number
of events in the actions and~$u$ is the number of updates, i.e., the 
number of actions.

As a subroutine, the algorithm checks whether a given basic epistemic
formula~$\phi$ holds in a given epistemic model~$M$, i.e.,
whether~$M \models \phi$.  
It is well-known that model checking for basic epistemic logic can be done in time polynomial in the of~$M$ plus the
size of~$\phi$ (see e.g.~\cite{blackburn2006}).

Let $x = (P, \A, i, s_0, a_1, \dots, a_f, \phi)$ be an instance of \gdbu{}. First the algorithm computes the final updated model~$s_f = s_0 \otimes a_1 \otimes \dots \otimes a_f $ 
by sequentially performing the updates. 
For each $i$, $s_i$ is defined as $s_{i-1} \otimes a_i$. The size of each $s_i$ is upper bounded by~$O(|s_0| \cdot e^u)$, 
so for each update checking the preconditions can be done in 
time polynomial in~$e^u \cdot |x|$. 
This means that computing $s_f$ can be done in 
fpt-time.%

Then, the algorithm decides whether~$ \phi$ is true in~$s_f$. This can be done in time polynomial in the size of $s_f$ plus the size of~$ \phi$. We know that~$|s_f| + |\phi|$ is upper bounded by {$O(|s_0| \cdot e^u) + |\phi|$}, thus upper bounded by~$e^u \cdot p(|x|)$, for some polynomial $p$.
Therefore, deciding whether~$\phi$ is true in~$s_f$ is fixed-parameter tractable.
Hence, the algorithm decides whether $x \in$ \gdbu{} and runs in fpt-time.
\end{proof}

\subsubsection{Overview of the Results}

\sloppypar
We showed that \gdbu{} is \pspace{}-complete, we presented several
parameterized intractability results ({\W{1}}-hardness and para-\np{}-hardness) and we presented one fixed-parameter tractable result, namely 
for \eugdbu{}. {In Figure~\ref{fig:overview}}, we present a graphical
overview of our {results} and 
the consequent border between fpt-tractability and fpt-intractability 
{for the problem} \gdbu{}. We leave \acpgdbu{} and \cfpugdbu{} as 
open problems for future research.

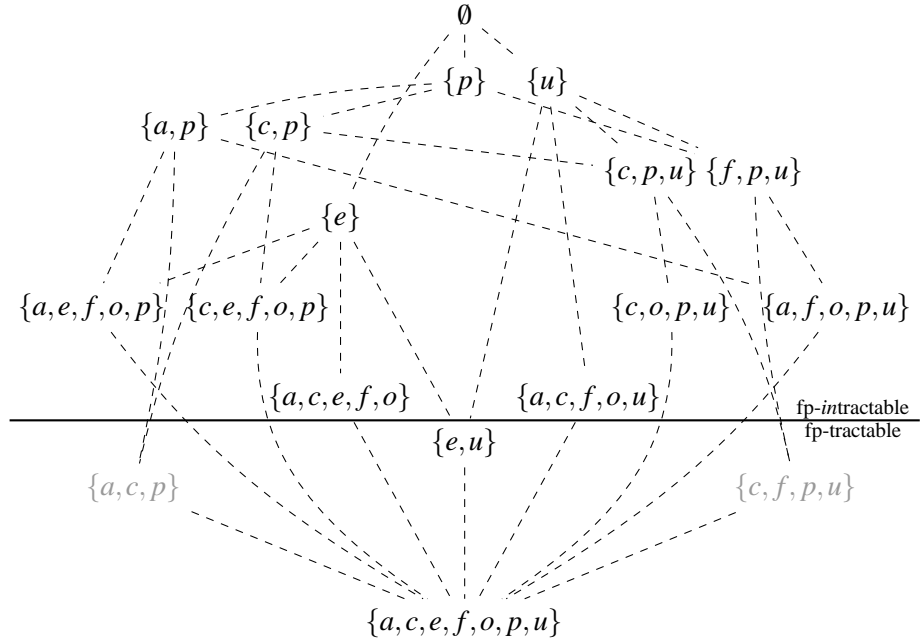
\begin{figure*}
\begin{center}
\begin{tikzpicture}[xscale=1.1, yscale=0.6]
  \tikzstyle{direct}=[draw,dashed]
  \tikzstyle{indirect}=[draw,dashed,dash pattern=on 1mm off 1mm,gray]
  \node[] (-) at (0,11.5) {$\emptyset$};
  \node[] (p) at (0,10) {$\SBs p \SEs$};
  \node[] (u) at (1,10) {$\SBs u \SEs$};
  \node[] (e) at (-1.5,7) {$\SBs e \SEs$};
  \node[] (eu) at (0,2) {$\SBs e,u \SEs$};
  \node[] (acfou) at (1.5,3) {$\SBs a,c,f,o,u \SEs$};
  \node[] (acefo) at (-1.5,3) {$\SBs a,c,e,f,o \SEs$};
  \node[] (ap) at (-3.5,9) {$\SBs a,p \SEs$};
  \node[] (cp) at (-2.25,9) {$\SBs c,p \SEs$};
  \node[] (fpu) at (3.5,8) {$\SBs f,p,u \SEs$};
  \node[] (cpu) at (2.25,8) {$\SBs c,p,u \SEs$};
  \node[] (aefop) at (-4.5,5) {$\SBs a,e,f,o,p \SEs$};
  \node[] (cefop) at (-2.5,5) {$\SBs c,e,f,o,p \SEs$};
  \node[] (afopu) at (4.5,5) {$\SBs a,f,o,p,u \SEs$};
  \node[] (copu) at (2.5,5) {$\SBs c,o,p,u \SEs$};
  \node[] (acp) at (-4,1) {\textcolor{black!40}{$\SBs a,c,p \SEs$}};
  \node[] (cfpu) at (4,1) {\textcolor{black!40}{$\SBs c,f,p,u \SEs$}};
  \node[] (acefopu) at (0,-2) {$\SBs a,c,e,f,o,p,u \SEs$};
  \path[direct] (-) edge[bend right=5] (e);
  \path[direct] (-) -- (u);
  \path[direct] (e) -- (eu);
  \path[direct] (u) -- (eu);
  \path[direct] (-) -- (p);
  \path[direct] (p) edge[bend right=10] (ap);
  \path[direct] (p) -- (cp);
  \path[direct] (p) -- (fpu);
  \path[direct] (cp) -- (cpu);
  \path[direct] (ap) -- (afopu);
  \path[direct] (e) -- (aefop);
  \path[direct] (e) -- (cefop);
  \path[direct] (u) -- (cpu);
  \path[direct] (u) -- (fpu);
  \path[direct] (ap) -- (aefop);
  \path[direct] (cp) -- (cefop);
  \path[direct] (fpu) -- (afopu);
  \path[direct] (cpu) -- (copu);
  \path[direct] (acp) -- (acefopu);
  \path[direct] (cfpu) -- (acefopu);
  \path[direct] (aefop) edge[bend right=10] (acefopu);
  \path[direct] (cefop) edge[bend right=20] (acefopu);
  \path[direct] (afopu) edge[bend left=10] (acefopu);
  \path[direct] (copu) edge[bend left=20] (acefopu);
  \path[direct] (ap) edge[bend left=3] (acp);
  \path[direct] (cp) edge[bend right=5] (acp);
  \path[direct] (cpu) edge[bend left=7] (cfpu);
  \path[direct] (fpu) edge[bend right=3] (cfpu);
  \path[direct] (eu) -- (acefopu);
  \path[direct] (e) -- (acefo);
  \path[direct] (u) -- (acfou);
  \path[direct] (acfou) -- (acefopu);
  \path[direct] (acefo) -- (acefopu);
  \path[draw,thick,black] (-5.5,2.5) -- (5.5,2.5);
  \node[] (fpt-label) at (4.7,2.25) {\scriptsize fp-tractable};
  \node[] (fpi-label) at (4.7,2.75) {\scriptsize fp-\emph{in}tractable};
\end{tikzpicture}
\end{center}
\caption{Overview of the parameterized complexity results
for the different parameterizations of DBU, and the line between
fp-tractability and fp-intractability
(under the assumption that the cases for~$\SBs a,c,p \SEs$
and~$\SBs c,f,p,u \SEs$ are fp-tractable).}
\label{fig:overview}
\end{figure*}

\section{Discussion \& Conclusions}
\label{discussion}

We presented the {\sc Dynamic Belief Update} {model} as a computational-level model of ToM
and analyzed its complexity. 
The aim of our model was to provide a formal approach that can be used to interprete and evaluate the meaning and veridicality of various complexity claims
in the cognitive science and philosophy literature concerning ToM. 
In this way, we hope to contribute to disentangling 
 debates in cognitive science and philosophy regarding the complexity of ToM.  

In Section~\ref{sectiongeneralcomplex}, we proved that \gdbu{} is \pspace{}-complete.
This means that (without additional constraints), there 
is no algorithm that computes \gdbu{} in a reasonable amount of time. 
In other words, without restrictions on its input domain,
the model is computationally too hard  to serve as a
plausible explanation for human cognition.
This may not be surprising, but it is a first formal proof 
backing up this claim, whereas so far claims of intractability 
in the literature remained informal.

Informal claims about what constitutes sources of intractability  abound in cognitive science. For instance, it seems to be folklore that the `order' of ToM reasoning (i.e., that I think that you think that I think \dots) is a potential source of intractability. 
The fact that people have difficulty understanding higher-order
theory of mind \cite{verbrugge2008,kinderman1998,lyons2010,stiller2007} is not explained by the complexity results for 
parameter~$o$ -- the modal depth of the formula that is being 
considered, in other words, the order parameter.
Already for a formula with modal depth one, \gdbu{} is NP-hard; 
so \ogdbu{} is not fixed-parameter tractable. 
On the basis of our results we can only conclude that \gdbu{} is fixed-parameter tractable for the order parameter in combination with 
parameters~$e$ and~$u$.
But since \gdbu{} is fp-tractable for the smaller 
parameter set~$\set{e,u}$, this does not indicate that the 
order parameter is a source of complexity.
This does not mean it may not be a source of difficulty for human ToM performance. After all, tractable problems can be too resource-demanding for humans for other reasons than computational complexity (e.g., due to stringent working-memory limitations).  

Surprisingly, we only found one (parameterized)
tractability result {for \gdbu}.
We proved that for parameter set $\set{e,u}$ -- the 
maximum number of events in an event model and the number of updates, 
i.e., the number of event models -- \dbu{} is fixed-parameter 
tractable. 
Given a certain instance~$x$  of \gdbu{}, the
values of parameters~$e$ and~$u$ (together with the size of initial 
state~$s_0$) determine the size of the final updated model (that 
results from applying the event models to the initial state). Small values 
of~$e$ and~$u$ thus make sure that
the final updated model does not blow up {too} much in relation to the 
size of the initial model.  
The result that \eugdbu{} is {fp-tractable} {indicates}
that the size of the final updated model {can be} 
a source of intractability (cf.~\cite{vanrooij2008analogical,vanRooij2008opportunities}).

The question arises how we can interpret  parameters~$e$ and $u$ 
in terms of their cognitive counterparts. To what aspect of ToM do 
they correspond, and moreover, can we assume that {they} have small
values in {(many)} real-life situations? If this is indeed the
case, then restricting the input domain of {the model}
to those inputs that have sufficiently small values
for parameters~$e$ and~$u$ will {render} our model 
tractable, and we can then argue that (at least in 
terms of its computational complexity) it is a cognitively plausible
model.

In his formalizations of the false belief task Bolander
\cite{bolander2014} indeed used a limited amount 
of actions with a limited amount of events in each action (he used
a maximum of 4). This {could,} however, be {a} consequence of
the over-simplification {(of real-life situations)} 
used in experimental tasks.
Whether these parameters in fact have sufficiently
small values in real life, 
is an 
empirical hypothesis that can (in principle) be tested 
experimentally.
 However, it is not straightforward how to interpret 
 these formal aspects of the model in terms of their
 cognitive counterparts.
 The associations that the words 
 \textit{event} and \textit{action} trigger with how we often 
 use these words in daily life, might adequately apply to some  
 degree, but could also be misleading. A structural way of 
 interpreting these parameters is called for. We think this
 is an interesting topic for future research.

Besides the role that our results play in
the investigation of (the complexity) of ToM
our results are also of interest in and of themselves.
The results in Theorems~\ref{thm4} and~\ref{thm5} 
resolve an open question in the literature about the computational
complexity of DEL.
Aucher and Schwarzentruber \cite{aucher2013} already showed that the model checking problem
for DEL, in general, is \pspace{}-complete. However, their proof for \pspace{}-hardness
does not work when the input domain is restricted to S5 (or KD45) models 
and their hardness proof also relies on the use of multi-pointed
models (which in their notation is captured by means of a union operator).
With our proof of Theorem~\ref{thm4}, we show that
DEL model checking is PSPACE-hard even when restricted to
single-pointed S5 models.
{Furthermore}, the novelty of our aproach lies in the fact
that we apply parameterized complexity analysis to dynamic 
epistemic logic, which is still a {rather unexplored area.}

\section*{Acknowledgements} 
We thank the reviewers for their comments. We thank Thomas Bolander, Nina Gierasimczuk, Ronald de Haan and Martin Holm Jensen and the members of the Computational Cognitive Science group at the Donders Centre for Cognition for discussions and feedback.

\DeclareRobustCommand{\VAN}[3]{#3}

\providecommand{\doi}[1]{\textsc{doi}: \href{http://dx.doi.org/#1}{\nolinkurl{#1}}}
\bibliographystyle{eptcs}

\newcommand{\noopsort}[1]{}

\end{document}